\documentclass{article}

\usepackage[margin=1.125in]{geometry}
\usepackage[dvipsnames]{xcolor}
\definecolor{DarkGreen}{rgb}{0.1,0.5,0.1}
\definecolor{DarkRed}{rgb}{0.5,0.1,0.1}
\definecolor{DarkBlue}{rgb}{0.1,0.1,0.5}

\definecolor{Lagunita}{RGB}{0,124,146}
\definecolor{BrightRed}{RGB}{177,4,14}
\definecolor{Mint}{RGB}{0, 155, 118}
\definecolor{Sky}{RGB}{0, 152, 219}

\usepackage{mleftright}

\usepackage[small]{caption}
\usepackage[pdftex]{hyperref}
\hypersetup{
    unicode=false,          
    pdftoolbar=true,        
    pdfmenubar=true,        
    pdffitwindow=false,      
    pdfnewwindow=true,      
    colorlinks=true,       
    linkcolor=DarkBlue,          
    citecolor=DarkGreen,        
    filecolor=DarkGreen,      
    urlcolor=DarkBlue,          
    pdftitle={},
    pdfauthor={},
}
\usepackage{amssymb,amsmath,amsthm,amsfonts,enumitem, dsfont}
\usepackage{fullpage,comment,mathtools}
\usepackage{tikz}
\usetikzlibrary{arrows,decorations.pathmorphing,decorations.shapes,patterns,positioning} 
\usepackage[ruled,linesnumbered]{algorithm2e}
\usepackage{thm-restate}
\mathtoolsset{showonlyrefs}

\newcommand{\F}{{\mathbb F}}

\newcommand{\tr}{\mathrm{tr}}

\newcommand{\n}{\newline}

\renewcommand{\epsilon}{\varepsilon}

\theoremstyle{plain}
\declaretheorem[name=Theorem,numberwithin=section]{theorem}
\declaretheorem[name=Lemma,sibling=theorem]{lemma}
\newtheorem*{lemma*}{Lemma}
\newtheorem*{theorem*}{Theorem}

\newtheorem{observation}[theorem]{Observation}

\newtheorem{cor}[theorem]{Corollary}
\newtheorem{corollary}[theorem]{Corollary}

\newtheorem{proposition}[theorem]{Proposition}

\newtheorem{fact}[theorem]{Fact}

\theoremstyle{definition}
\newtheorem{definition}[theorem]{Definition}

\newcommand{\eval}{\mathrm{eval}}

\DeclarePairedDelimiterX\set[1]\lbrace\rbrace{\def\given{\;\delimsize\vert\;}#1}

\colorlet{LightSky}{Sky!30}

\title{Wedge-Lifted Codes}
\author{
Jabari Hastings\thanks{Department of Mathematics, Stanford University, Research supported by CURIS 2020. \texttt{jabarih@stanford.edu}} ,
Amy Kanne\thanks{Department of Mathematics, Stanford University, Research supported by CURIS 2020. \texttt{akanne@stanford.edu}} ,
Ray Li\thanks{Department of Computer Science, Stanford University.  Research supported by an NSF GRFP under grant DGE - 1656518. \texttt{rayyli@stanford.edu}} ,
Mary Wootters\thanks{Departments of Computer Science and Electrical Engineering, Stanford University. This work is partially supported by NSF grants CCF-1657049 and CCF-1844628. \texttt{marykw@stanford.edu}}}

\begin{document}

\newcommand{\codename}{wedge-lifted code }
\newcommand{\codenamens}{wedge-lifted code} 
\newcommand{\Codename}{Wedge-lifted code }
\newcommand{\Codenamens}{Wedge-lifted code} 
\newcommand\ddd{d}

\maketitle
\begin{abstract}
    We define \emph{wedge-lifted codes}, a variant of lifted codes, and we study their locality properties.  We show that (taking the trace of) wedge-lifted codes yields binary codes with the $t$-disjoint repair property ($t$-DRGP).  When $t = N^{1/2\ddd}$, where $N$ is the block length of the code and $\ddd \geq 2$ is any integer, our codes give improved trade-offs between redundancy and locality among binary codes.
\end{abstract}
\section{Introduction}\label{sec:intro}

In this work, we define and study \emph{Wedge-Lifted Codes}, and show they yield improved binary error correcting codes for a notion of locality known as the $t$-disjoint repair group property.

An \emph{error correcting code} (or simply \emph{code}) $\mathcal{C}\subset\Sigma^N$ is a set of strings of a fixed \emph{length} $N$ over an \emph{alphabet} $\Sigma$.
If $\Sigma=\{0,1\}$, $\mathcal{C}$ is called a \emph{binary} code.
We measure the quality of a code by the \emph{redundancy}, denoted $K^{\perp}$, defined as $K^{\perp}=N-K$, where $K=\log_{|\Sigma|}|\mathcal{C}|$ is the \emph{dimension} of the code.
It is desirable for codes to be larger, or equivalently to have less redundancy.

In this work we are interested in constructing better (less redundant) binary codes with \emph{locality}.
There are several notions of locality in this literature, but informally a code exhibits locality if we can correct one or a small number of erasures by looking only \em locally \em at a few other symbols of the codeword. 
In this work we construct codes with a notion of locality known as the $t$-disjoint repair group property ($t$-DRGP).
\begin{definition} 
A code $\mathcal{C}\subseteq \Sigma^N$ has the $t-$\textit{disjoint repair group property} ($t$-DRGP) if for every $i \in [N]$, there is a collection of $t$ disjoint subsets $S_1,\ldots,S_t\subseteq [N]\setminus\{i\}$ and functions $f_1,\ldots,f_t$ so that for all $c \in \mathcal{C}$ and $j \in [t]$, $f_j\left(c|_{S_j}\right) = c_i$.
\end{definition}

Codes with the $t$-DRGP are motivated by distributed storage, where one desires efficient recovery from a few erasures (see surveys \cite{RSGKBR13,Ska16}). 
Codes with the $t$-DRGP are also relevant to \emph{private information retrieval} (PIR) in cryptography, as all codes (with a systematic encoding) with the $t$-DRGP also form \emph{$(t+1)$-PIR} codes, a notion defined in \cite{FVY15} for $(t+1)$-server private information retrieval.  Further, when $t=\Omega(N)$ is large, codes with the $t$-DRGP are equivalent to \emph{locally correctable codes}~\cite{KT00,Woo10}.

Previously the best constructions for codes with the disjoint repair group property were given by lifted multiplicity codes \cite{LW19, HPPVY20} (see also \cite{Wu15}), but these codes have very large alphabet sizes, which is undesirable from the perspective of distributed storage and PIR. 
Hence, a natural question, explicitly asked in \cite{LW19}, is, what are the best constructions of \emph{binary} codes with the $t$-DRGP?
Our work makes progress on this question by giving new constructions of binary codes with the $t$-DRGP with the best known redundancy for some values of $t$.
\begin{theorem}
\label{thm:main}
\label{thm:main-0}
For positive integers $\ddd$ and infinitely many $N$, for $t=N^{1/2\ddd}$, there exist binary codes of length $N$ with redundancy $t^{\log_2(2-2^{-\ddd})}\sqrt{N}$ that have the the $(t-1)$-DRGP.
\end{theorem}
Theorem~\ref{thm:main-0} gives improved constructions of binary codes with the $t$-DRGP when $t=N^{1/2\ddd}$ for integers $\ddd\ge 2$. (see Figure~\ref{fig:litreview} for a visual comparison and Section~\ref{sec:related} for a more detailed comparison):
When $t=N^{1/4}$, Theorem~\ref{thm:main-0} improves over a construction of \cite{FischerGW17}, and for $t=N^{1/2\ddd}$ for $\ddd\ge 3$, Theorem~\ref{thm:main-0} improves over the constructions of \cite{FVY15}.
As all codes with the $(t-1)$-DRGP are also $t$-PIR codes, Theorem~\ref{thm:main-0} also gives improved constructions of $t$-PIR codes in the same parameter settings.

Our Wedge-Lifted Codes have relatively small alphabet size (compared to previous work using lifted multiplicity codes, for example), and so we are able to use them to obtain good binary codes by taking the coordinate-wise trace over $\mathbb{F}_2$.

It is an interesting question whether the construction of \cite{FVY15} can be beaten for all $t\in(1,\sqrt{N})$.
That is, for all $t=N^{\alpha}$ with $\alpha\in(0,1/2)$, are there binary codes with the $t$-DRGP and redundancy $O(t^{1-\varepsilon}\sqrt{N})$ for some $\varepsilon>0$ (possibly depending on $\alpha$)?
Our work shows this is true for $\alpha=1/2\ddd$ when $\ddd$ is a positive integer.
We additionally show that, for a dense collection of $\alpha\in(0,1/2)$, our binary codes essentially match the redundancy bound of $O(t\sqrt{N})$ from \cite{FVY15} (see Theorem~\ref{thm:fvy}).
This is proved with a naive bound, so it is possible that, with a more refined analysis, our codes could achieve the improvement to $O(t^{1 - \varepsilon}\sqrt{N})$ for all $\alpha\in(0,1/2)$.

\paragraph{Organization}
In the remainder of this section, we highlight some related work and our approach.
In Section~\ref{sec:prelim}, we state some preliminaries for our work.
In Section~\ref{sec:wedge}, we define and analyze our construction of Wedge-Lifted Codes. 
In Section~\ref{sec:trace}, we show demonstrate how to turn the codes in Section~\ref{sec:wedge}, which are over a $q$-ary alphabet, into binary codes, proving Theorem~\ref{thm:main}.

\begin{figure}[!ht]
\centering
\begin{tikzpicture}[xscale=30, yscale=30]

\draw[->] (0,.48) to (0,.81);
\node[anchor=south] at (0,.82) {$\log_N( K^\perp )$};

\node[anchor=east] at (-.01, .5) {\footnotesize $.500$};
\node[anchor=east] at (-.01, .619) {\footnotesize $.619$};
\node[anchor=east] at (-.01, .651) {\footnotesize $.651$};
\node[anchor=east] at (-.01, .702) {\footnotesize $.702$};
\node[anchor=east] at (-.01, .714) {\footnotesize $.714$};
\node[anchor=east] at (-.01, .75) {\footnotesize $.750$};
\node[anchor=east] at (-.01, .792) {\footnotesize $.792$};

\draw[->] (0,.48) to (0.3,.48);
\node[anchor=west] at (0.3,0.48) {$\log_N(t)$};
\node at (1/8, .45) {$\frac{1}{8}$};
\node at (1/6, .45) {$\frac{1}{6}$};
\node at (1/4, .45) {$\frac{1}{4}$};

\draw[thick,Sky, domain=0:.25, dashed, variable=\x] plot ( {\x}, {\x + .5} ); \node[anchor=west](fvy-label) at (.26, 0.75) {\footnotesize \cite{FVY15}};

\draw[thick, dashed, Sky] (0, .792) to (0.25, 0.792) {};
\node[anchor=west](gks-label) at (.26, .792) {\footnotesize \cite{GuoKS13}};

\node[draw, fill=Sky, circle, scale=0.3](fgw) at (1/4, .714) {};
\node[anchor=west](fgw-label) at (.26, .714) {\footnotesize
\cite{FischerGW17}};

\draw[dotted, thick, violet] (0,.5) to (0.25, .5);
\node[anchor=west, label={[align=center]\footnotesize \cite{Woo16,RV16}\\\footnotesize lower bound}](lbound-label) at (.3, .49) {};

\draw[domain=0:.25, dashed, variable=\x, LightSky, thick] plot ( {\x} , {.5 + \x * .585});
\node[anchor=west](lm-label) at (.26, .65) {\footnotesize \cite{LW19} (big $\Sigma$) };

\draw[domain=0:0.25, variable=\x, Sky, dashed, thick] plot ( { \x }, {(3/4 + \x*(log2(3/8) + 3/2}  );

\node[anchor=west](lm-label) at (.26, .77) {\footnotesize \cite{LW19} (binary) };

\foreach \n in {2,...,30} {
  \node[draw, fill=BrightRed, circle, scale=0.3](w\n) at ({1/(2*\n)}, {0.5 + log2(2-pow(2,-\n))/(2*\n)}) {};
}
\node[anchor=west](this-label-2) at (.26, .702) {\footnotesize \color{BrightRed}This work};

\draw[dotted, black](1/6, 0.48) to (w3);
\draw[dotted, black](1/8, 0.48) to (w4);
\draw[dotted, black](1/4, 0.48) to (1/4, 0.792);

\draw[dotted, black] (fgw) to (0, .714);
\draw[dotted, black] (w2) to (0, .702);
\draw[dotted, black] (w3) to (0, .651);
\draw[dotted, black] (w4) to (0, .619);

\end{tikzpicture}

\caption{The best trade-offs known between the number $t$ of disjoint repair groups and the redundancy $K^\perp$, for $t \leq N^{1/4}$. Our results appear as the red dots.}

\label{fig:litreview}

\end{figure}
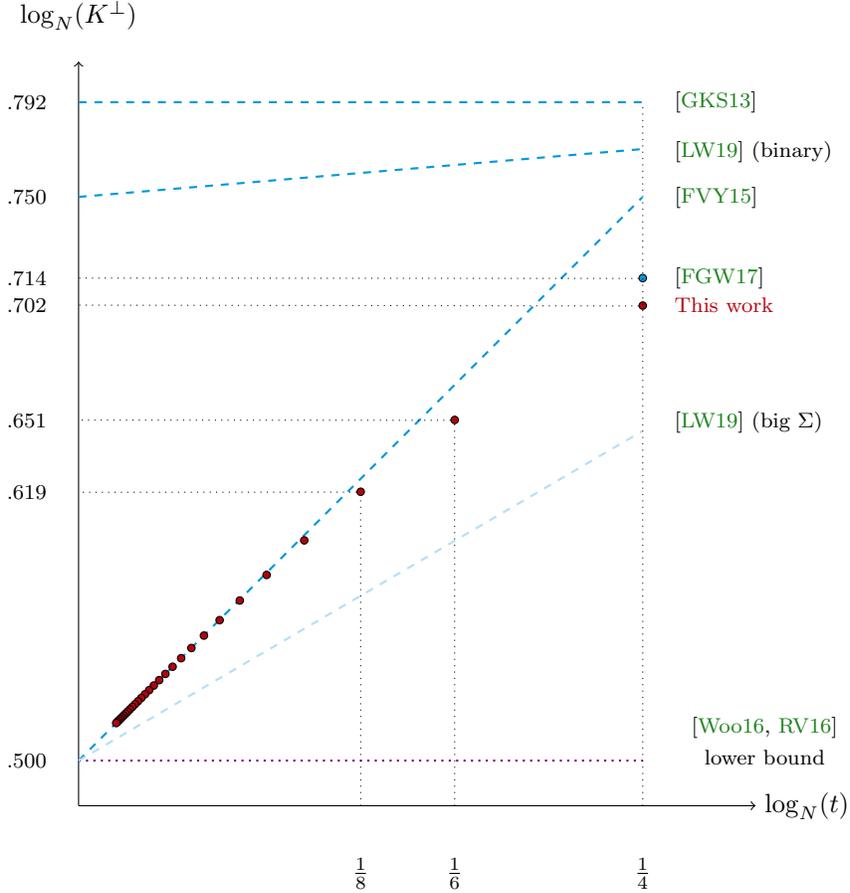

\subsection{Related work}
\label{sec:related}

\paragraph{Prior work on DRGP}
We refer the reader to Figure~\ref{fig:litreview} for a picture of the best known constructions of codes with the $t$-DRGP.
For $t\le \sqrt{N}$, \cite{FVY15} constructed codes with $t$-DRGP and $t\sqrt{N}$ bits of redundancy.
We show (Corollary~\ref{cor:fvy}) that there always exist Wedge-Lifted Codes that \emph{at least} match the redundancy this construction.
However, the construction of \cite{FVY15} is not optimal, at least when $t$ is somewhat large.
\cite{GuoKS13} showed that \emph{Lifted Reed Solomon} codes have the $N^{1/2}$-DRGP with redundancy $N^{\log_4(3)}\approx N^{0.792} \ll t\sqrt{N}$.
A later work \cite{FischerGW17} also improved \cite{FVY15} when $t=N^{1/4}$, giving codes with redundancy at most $N^{0.714}$.
In the setting where $t=N^{1/4}$, Theorem~\ref{thm:main} improves on \cite{FischerGW17} by giving codes with redundancy roughly $N^{0.702}$, and for $t=N^{1/2\ddd}$ for all $\ddd\ge 3$, Theorem~\ref{thm:main-0} improves \cite{FVY15} by giving codes with redundancy $t^{1-\varepsilon}\sqrt{N}$ for some $\varepsilon>0$ depending on $\ddd$.

When $t > N^{1/2}$, \cite{AY17} showed that for all $\delta>0$ there exist codes with the $t$-DRGP when $t = N^{1-\delta}$ and with redundancy $O(N^{1-\varepsilon})$ for some $\varepsilon>0$ depending on $\delta$. The best tradeoffs between $\delta$ and $\varepsilon$ are given in \cite{HPPVY20}, which refines the analysis of \cite{GuoKS13}.
Our work focuses on the setting where $t \le N^{1/2}$ because we believe this parameter regime is already interesting, and because we believe that the improvements over existing binary codes given by Wedge-Lifted Codes are the strongest in this parameter regime.
By generalizing our bivariate construction of Wedge-Lifted Codes to more than two variables, we expect it is possible to obtain improved binary codes with the $t$-DRGP when $t > N^{1/2}$.

\paragraph{Other notions of locality}
The $t$-DRGP is closely related to other notions of locality.
Any linear code with the $t$-DRGP is by definition a $(t+1)$-private information retrieval (PIR) code.
In PIR codes, one only needs to recover the information symbols, rather than all codeword symbols \cite{FVY15, BE16, AY17}.
Our constructions are linear and therefore immediately give constructions of $(t+1)$-PIR codes.
The $t$-DRGP is also closely related to \emph{batch codes}, which generalize PIR codes \cite{IKOS04, DGRS14, AY17, HPPV20}.
The $t$-DRGP is also related to \emph{locally correctable codes} (LCCs) \cite{GuoKS13, KSY14, HOW13, KMRS16} and \emph{locally decodable codes} (LDCs) \cite{KT00, E13}.
Any code of length $N$ has the $\Omega(N)$-DRGP if and only if it is a constant query LCC, and any code of length $N$ is an $\Omega(N)$-PIR code if and only if it is a constant-query LDC.
This equivalence between LCCs/LDCs and disjoint repair groups was in fact used to prove lower bounds for LCCs/LDCs \cite{KT00, Woo10}.
Towards applications in distributed storage, a generalization of the $t$-DRGP called \emph{LRCs with availability} have also been studied \cite{WZ14, RPDV14, TB14, TBF16}, where the disjoint repair groups are additionally constrained to have a bounded size.

\subsection{Our approach}
The best known constructions of codes \cite{GuoKS13, FischerGW17, LW19, HPPVY20} with the DRGP leverage an idea called \emph{lifting} \cite{GuoKS13}.
The basic idea of lifting is to improve the Reed-Muller Code, which has some locality properties, by relaxing the construction so that the locality properties are preserved but so that the redundancy decreases.
To illustrate lifting and how we adapt the idea to our work, we assume for the rest of the discussion that we are working with bivariate polynomial codes over a field $\mathbb{F}_q$ of characteristic 2.
In this way, codewords can be viewed not only as elements of $\mathbb{F}_q^{q^2}$, but also as (evaluations of) bivariate polynomials in $\mathbb{F}_q[X,Y]$ of degree at most $q-1$ in each variable.
While lifted codes can be analyzed over larger field characteristics, lifted codes seem to obtain the best redundancy over fields characteristic 2, so we focus our attention on characteristic 2.

The bivariate Reed-Muller Code over $\mathbb{F}_q$ obtained by evaluating polynomials of total degree $\le q-2$ has length $N=q^2$, has redundancy roughly $N/2$, and has the $\sqrt{N}$ disjoint repair group property, because each line forms a repair group and there are $\sqrt{N}$ lines in $\mathbb{F}_q^2$ passing through each point.
Indeed, each line forms a repair group because the line-restriction of any bivariate polynomial with total degree at most $q-2$ is a univariate polynomial of degree at most $q-2$, and thus the sum of the evaluations along any line is 0.
Guo Kopparty and Sudan \cite{GuoKS13} show that, perhaps surprisingly, if we remove the total degree condition but still require lines to form repair groups, we can decrease the redundancy of the code from $\Omega(N)$ to $N^{0.792}$ (if the field $\mathbb{F}_q$ has characteristic 2).  

To obtain codes with the $t$-DRGP for values of $t$ other than $\sqrt{N}$, variations of this idea have been proposed.
In \cite{FischerGW17}, the authors define \emph{partially lifted codes}, requiring that only \emph{some}, rather than all, lines in $\mathbb{F}_q^2$ form repair groups, and show this gives an improved construction for $t=N^{1/4}$.
The works of \cite{Wu15, LW19, PV19, HPPVY20} adapt \cite{GuoKS13} by adding derivatives to the construction to obtain \emph{lifted multiplicity codes}, and show these give the best known constructions for the DRGP for all values of $t\in(1,N^{0.99})$ (but with a large alphabet).
In \cite{LMMPW20}, the authors consider another variant, \emph{Hermitian lifted codes} where the polynomials are evaluated on a Hermitian curve rather than on all of $\mathbb{F}_q^2$.

In this work, we propose another method of adapting the construction of \cite{GuoKS13} to other values of $t$, which we call \emph{wedge-lifts}.
This method avoids the large alphabets of lifted multiplicity codes \cite{Wu15, LW19, PV19, HPPVY20}, improves on the partially lifted codes of \cite{FischerGW17}, and additionally gives the best constructions of codes with the $t$-DRGP for $t=N^{1/2\ddd}$ for integers $\ddd\ge2$.
Roughly, the idea of wedge-lifts is that, instead of requiring that the repair groups are lines, we take the repair groups to be sets of lines (``wedges'') through a point.
For a line to form a repair group in the ordinary lifted code, we need the sum of the evaluations along the line to be 0.
For a wedge to form a repair group in the wedge-lifted code, we require the sum of all the points in the wedge to be 0.
If the wedges have odd size (the field has characteristic 2), then each ``wedge parity check'' is the sum of ``line parity checks'', and thus the ordinary lifted code is a subset of the wedge-lifted code.
We then show that, if the wedges are formed by lines lying in a coset of some subgroup $H\le \mathbb{F}_q^\times $ of the multiplicative group of $\mathbb{F}_q$, then the codes have redundancy at most $t\sqrt{N}$, matching the constructions in \cite{FVY15}.
Furthermore, for certain choices of the subgroup $H$, we can analyze the redundancy of the construction more precisely, giving the improved bounds on the redundancy in Theorem~\ref{thm:main}.
We are only able to precisely analyze wedge-lifted codes for some choices of the subgroup $H$, and we leave it to future work to give better bounds on the redundancy for other subgroups $H$.

Our constructions of wedge-lifted codes are $q$-ary codes.  In order to construct binary codes, as stated in Theorem~\ref{thm:main}, we take the (coordinate-wise) \emph{trace} of the code.   This technique of taking the trace, which is implicit in \cite{GuoKS13}, saves a $\log(q)$ multiplicative factor in the redundancy over a simpler technique of reducing the alphabet size, which is used in \cite{AY17, LW19, HPPVY20}.

\section{Preliminaries}
\label{sec:prelim}

\noindent In this section, we introduce the background and notation we will use throughout the paper.
\subsection{Notation and basic definitions}

\noindent Let $\F_q$ denote the finite field of order $q$ and let $\F_q^\times$ denote its multiplicative subgroup. We study linear codes $\mathcal C \subseteq \F_q^N$ of block length $N$ over an alphabet of size $q$. Throughout this paper, we assume that $\F_q$ has characteristic $2$ and write $q = 2^\ell$. 

We need the following tools to reason about the binary representations of integers.
\begin{definition}
Let $x,y$ be two non-negative integers with binary representations $x = \overline{x_{\ell-1}\dotsm x_0}, y = \overline{y_{\ell-1}\dotsm y_0}$. For each $j \in \{0,\ldots,\ell-1\}$, let $a_j = \max(x_j,y_j)$ and $b_j = \min(x_j,y_j)$. Then, define the bitwise-OR $\vee$ and bitwise-AND $\wedge$ of $x$ and $y$ as follows
\[x\vee y = \overline{a_{\ell-1}\dotsm a_0}\]
\[x\wedge y = \overline{b_{\ell-1}\dotsm b_0}\]
\end{definition}

\begin{definition}
For non-negative integers $x$ and $y$ with binary representations (possibly with leading zeros) $x = \overline{x_{\ell - 1} \dotsm x_0}, y = \overline{y_{\ell - 1} \dotsm y_0}$, respectively, we say that $x$ \emph{lies in the $2$-shadow} of $y$, denoted $x \leq_2 y$, if $x_i \leq y_i$ for all $i \in \{0, \dots, \ell- 1 \}$.
\end{definition}

We are interested in $2$-shadows because of Lucas's theorem.
\begin{theorem}[Lucas's Theorem]\label{thm:lucas}
Let $p$ be a prime and $x = \overline{x_{\ell-1}\dotsm x_0}, y = \overline{y_{\ell-1}\dotsm y_0}$ be written in base $p$. Then,
\[\binom{x}{y} = \prod_{i=0}^{\ell-1} \binom{x_i}{y_i} \mod p\]
In the case where $p=2$, then $\binom{x}{y} = 1 \mod p$ if and only if $y \leq_2 x$.
\end{theorem}

The codes $\mathcal C$ we consider are  polynomial evaluation codes. 
For a polynomial $P \in \F_q[X_1, \dots, X_m]$, we write its corresponding codeword as 
\[
\eval_q(P) = \left < P(x_1,\dots x_m) \right >_{(x_1, \dots, x_m) \in \F_q^m}. 
\]
Above, we assume some fixed order on the elements of $\F_q^m$. In this paper, we focus on when $m$ is $1$ or $2$. We are concerned with the restriction of bivariate polynomials to lines.
\begin{definition}
For a line $L : \F_q \to \F_q^2$ with $L(T) = (L_1(T) , L_2(T))$ and a polynomial $P : \F_q^2 \to \F_q$, we define the restriction of $P$ on $L$, denoted $P|_L$, to be the unique polynomial of degree at most $q -1$ so that $P|_L(T) = P(L_1(T), L_2(T))$.
\end{definition}

\subsection{Trace codes}

Given a linear code $\mathcal C \subseteq \F_{p^\ell}^N$, it is sometimes desirable to construct another code $\mathcal C' \subseteq \F_p^N$ over a smaller alphabet that maintains some properties of $\mathcal C$. We briefly describe a way to create such a code using the trace function, $\tr_p$.
\begin{definition}
Let $\tr_p: \F_{p^\ell} \to \F_p$ be the trace function
\[\tr_p(\alpha) = \sum_{i=0}^{\ell-1} \alpha^{p^i}\]
We can extend $\tr_p:\F_{p^\ell}^N \to \F_p^N$ by defining
\[\tr_p(v) = (\tr_p(v_1),\ldots,\tr_p(v_N)) \in \F_p^N\]
We can further extend it to codes $\mathcal C\subseteq \F_{p^\ell}^N$ by taking the the trace of every vector in the code:
\[\tr_p( \mathcal C) = \{\tr_p(v) : v \in \mathcal C\} \subseteq \F_p^N\]
\end{definition}

Note that $\tr_p$ is a $\F_p$-linear map. Hence, if $\mathcal C$ is a linear code, then $\tr_p (\mathcal{C})$ is also a linear code. We can bound the rate of $\tr_p(\mathcal{C})$ in terms of the rate of $\mathcal C$ using the following corollary of Delsarte's theorem.
\begin{theorem}[Follows from Delsarte's theorem (see, e.g., \cite{HS09})] \label{thm:delsarte}
For any $\F_{p^\ell}$-linear code $\mathcal C \subseteq \F_{p^\ell}^N$,
\[
\dim \mathcal C \leq \dim \tr_p(\mathcal C) \leq \ell \cdot \dim \mathcal C
\]
\end{theorem}

\section{\Codenamens{s}}
\label{sec:wedge}
In this section, we define and analyze wedge-lifted codes. As mentioned above, we focus on bivariate codes, although as remarked above we believe our work could be extended to more variables.  To that end, from now on we use $\eval_q$ denote the bivariate evaluation map:  
\begin{definition}
Let $\eval_q: \F_q[X,Y] \to \F_q^{q^2}$ be the evaluation map of the polynomial at all of $\F_q^2$, namely $\eval_q(P) := (P(\textbf{x}))_{\textbf{x} \in \F_q^2}$
\end{definition}
In Section~\ref{ssec:wedge}, we give the formal definition of wedge-lifted codes.  In Section~\ref{ssec:cosets}, we analyze wedge-lifted codes in a special case where the construction arises from cosets of $\mathbb{F}_q^{\times}$.  Finally in Section~\ref{ssec:instantiations}, we further set parameters to fix particular coset decompositions, culminating in our main theorem of the section, Theorem~\ref{thm:q-ary-wedge}.

\subsection{Definition of \codenamens{s}}\label{ssec:wedge}
Finally we are ready to define wedges, wedge-restrictions, and  \codenamens{s}. 
\begin{definition}[Wedge]
For a point $\vec{p} =(x, y) \in \F_q^2$ and a set $H \subseteq \F_q$, we define the
\emph{wedge} through $\vec{p}$ formed by $H$, denoted $W_{H, \vec{p}}$, as the set of affine lines passing through $\vec{p}$ whose slope is in $H$,
\[
W_{H, \vec{p}} =  
\set*{
(L_1(T), L_2(T) ) \given 
\begin{array}{c}
L_1(T) = T \  \text{and  for some } \alpha \in H, \\
L_2(T) = \alpha(T - x) + y 
\end{array}
}.
\]
\end{definition}
\begin{definition}[Wedge restriction]
For a wedge $W_{H, \vec{p}}$ and polynomial $P \in \F_q[X, Y]$, we define the \emph{wedge restriction} of $P$ to the wedge $W_{H, \vec{p}}$, denoted $P|_{W_{H, \vec{p}}}$, to be the sum of the restrictions of $P$ to each line in the wedge,
\[
P|_{W_{H, \vec{p}}} =  \sum_{L \in W_{H, \vec{p}}} P|_L =  \sum_{\alpha \in H} \sum_{T \in \F_q} P(T, \alpha (T - x) + y).
\]
\end{definition}

Note that when $W_{H, \vec{p}}$ consists of an odd number of lines and the field has characteristic $2$, the wedge-restriction of $P$ is equivalent to the sum of $P$'s evaluations of each point in the wedge,
\begin{align*}
 P|_{W_{H, \vec{p}}} = \sum_{ {(x, y) } \in W_{H, \vec{p}}} P(x, y).
\end{align*}

\begin{definition}[\Codenamens{s}]\label{def:WLC}
Let $\mathcal H$ be a collection of disjoint subsets of $\F_q$, with each $H \in \mathcal H$ having odd size. The $(\mathcal H, q)$ \emph{wedge-lifted code} is a code $\mathcal C$ over alphabet $\Sigma = \F_q$ of length $q^2$ given by 
\begin{align}
\mathcal C = 
\set*{ 
\eval_q(P)
\given
\begin{array}{c}
P \in \F_q[X, Y] \ \text{and}, \ \text{for any}  \ H \in \mathcal H, \\
\text{and any} \ \vec{p} \in \F_q^2, P|_{W_{H, \vec{p}}} = 0  
\end{array} 
}.
\end{align}
\end{definition}

Following the approach of previous work \cite{GuoKS13,FischerGW17, LW19, PV19, HPPVY20}, we show that \codenamens{s} contain the evaluations of polynomials that lie in the span of ``good'' monomials. Informally, a monomial is $(\mathcal H, q)$-good if it restricts nicely to all the wedges we can construct with $\mathcal H$.

\begin{definition}[$(\mathcal H, q)$-good monomials]
Let $\mathcal H$ be a collection of disjoint subsets of $\F_q$.
We say that a monomial $P(X, Y) = X^aY^b$ with $0 \leq a, b \leq q - 1$ is $(\mathcal H, q)$-\emph{good} if for every $H \in \mathcal H$ and every $\vec{p} \in \F_q^2$, $P|_{W_{H, \vec{p}}}  = 0$, and say it is $(\mathcal H, q)$-\emph{bad} otherwise.
\end{definition}

By definition, the evaluations of all good monomials lie within our \codenamens{s}. Furthermore, monomials $X^aY^b$ with $a,b \leq q-1$ form a basis for polynomials of degree at most $q-1$, which are in bijection with $\F_q^{q^2}$ through the $\eval_q$ map. Therefore, we can obtain a lower bound on the rate of our code by finding a lower bound on the number of good monomials.

\begin{observation}\label{lem:rate}
For any $(\mathcal H, q)$ \codename $\mathcal C$, the redundancy of $\mathcal{C}$ is at most the number of $(\mathcal H , q)$-bad monomials.  
\end{observation}

Moreover, the number of disjoint repair groups for a wedge-lifted code is straightforward to compute. 

\begin{proposition}
\label{prop:repaircount}
Any $(\mathcal H,q)$ wedge-lifted code $\mathcal{C}$ has $|\mathcal H|$ disjoint repair groups.
\end{proposition}

\begin{proof} For any $\vec{p} \in \F_q^2$ and $H \in \mathcal H$, we can recover $P(\vec{p})$ by summing $P$ along $W_{H,\vec{p}}\setminus \{\vec{p}\}$. This sum gives $P(\vec{p})$ exactly because $\F_q$ is characteristic two. So, each $H$ admits one repair group for any $\vec{p}$, namely $W_{H,\vec{p}}\setminus \{\vec{p}\}$. The fact that the $H \in \mathcal H$ are disjoint gives us that these repair groups are disjoint too, because non-parallel lines intersect at only one point, $\vec{p}$.
\end{proof}

\subsection{\Codenamens{s} via cosets}\label{ssec:cosets}
In this section we analyze \codenamens{s} that arise when $\mathcal H$ is a collection of cosets. We work with cosets because they conveniently partition $\F_q^\times$ into sets of odd size when $q$ is a power of 2, which in turn gives us the $|\mathcal H|$-disjoint repair group property. Moreover, we can make use of the following fact, which simplifies our analysis.
\begin{fact}
\label{fact:sum}
Let $H \leq \F_q^\times$ be a subgroup. For any nonnegative integer $n$, 
\[ \sum_{\alpha \in H} \alpha^n = \begin{cases} |H| & n \equiv 0 \mod |H| \\ 0 & \text{otherwise} \end{cases}. \]
\end{fact}

\begin{lemma}\label{lem:bad-characterization}
Let $H\leq \F_q^\times$ be a subgroup, $\mathcal H$ be the collection of cosets $g H$ of $\F_q^\times$, and $a,b$ be integers such that $0 < a,b\leq q-1$ and $a+b<2(q-1)$. Then, a monomial $X^aY^b$ is $(\mathcal H,q)$-bad if and only if both of the following conditions hold:
\begin{enumerate}
    \item $a\vee b = q-1$
    \item There exists an $i \equiv b \mod |H|$ such that $i \leq_2 a\wedge b$.
\end{enumerate}
\end{lemma}
\begin{proof}
Say that $\vec{p} = (x,y)$ for any $x,y \in \F_q$, we aim to show that $P|_{W_{g H, \vec{p}}} = 0$ for any choice of $x,y$ if and only if the two conditions hold. We begin by simplifying $P|_{W_{g H, \vec{p}}}$ in the case that $g H$ is a coset of a subgroup $H \leq F_q^\times$.

\begin{equation}\label{eq:expansion}
    P|_{W_{g H, \vec{p}}} = \sum_{\alpha \in gH}\sum_{T \in \F_q}\sum_{i=0}^b \binom{b}{i} \alpha^i(y-\alpha x)^{b-i} T^{a+i} = \sum_{\alpha \in H}\sum_{T \in \F_q}\sum_{i=0}^b \binom{b}{i} (g\alpha)^i(y-g\alpha x)^{b-i} T^{a+i}
\end{equation}
First we handle the case where $a+b < 2q-2$. With this assumption and the fact that $\sum_{T \in \F_q} T^{a+i} = 1$ if $a+i = q-1$ and $0$ otherwise, we can simplify the above sum.
\[ P|_{W_{g H, \vec{p}}} = \sum_{\alpha \in H} \binom{b}{q-1-a} (g\alpha)^{-a}(y-g\alpha x)^{b+a-(q-1)} \]

Theorem \ref{thm:lucas} 
states that $\binom{b}{q-1-a} = 1$ if and only if $q-1-a \leq_2 b$, which is equivalent to $a\vee b = q-1$. So, if $a\vee b \neq q-1$, $P|_{W_{g H, \vec{p}}} = 0$, so $P$ is good. 

It remains to show that if $a\vee b = q-1$, then $P|_{W_{\alpha H, \vec{p}}} = 0$ if and only if there exists an $i \equiv b \mod |H|$ such that $i \leq_2 a\wedge b$. We then expand the binomial term in the sum of the above expression.
\begin{align*}
    P|_{W_{g H, \vec{p}}} &= \sum_{\alpha \in H} \binom{b}{q-1-a} (g\alpha)^{-a}(y-g\alpha x)^{a+b-(q-1)} \\
    &= \sum_{\alpha \in H} \sum_{i=0}^{b+a-(q-1)}\binom{a+b-(q-1)}{i}y^ix^{a+b-(q-1)-i}(g\alpha)^{b-i}\\
    &=  \sum_{i=0}^{b+a-(q-1)}\binom{a+b-(q-1)}{i}y^ix^{a+b-(q-1)-i}g^{b-i}\sum_{\alpha \in H} \alpha^{b-i}\\
    &=  \sum_{\substack{0 \leq i \leq b+a-(q-1)\\i \equiv b \mod |H|} }^{b+a-(q-1)}\binom{a+b-(q-1)}{i}y^ix^{a+b-(q-1)-i}g^{b-i}
\end{align*}
We use Fact \ref{fact:sum} to simplify from the third line to the fourth. Then, if we view this as a bivariate polynomial in $x,y$, we see that its total degree is at most $q-1$, so it is $0$ for every $x,y \in \F_q$ if and only if every coefficient of $x$ and $y$ is $0$.
Therefore, $P$ is bad if and only if $g^{b-i}\binom{a+b-(q-1)}{i} \neq 0$ for some $i \equiv b \mod |H|$. Because $a\vee b = q-1$ in this case, we can compute that $a+b-(q-1) = a\wedge b$. Moreover, $g\neq 0$, so applying Theorem \ref{thm:lucas} to $\binom{a\wedge b}{i}$ gives us the following equivalence.
\[g^{b-i}\binom{a+b-(q-1)}{i} \neq 0 \Longleftrightarrow i \leq_2 a\wedge b,\]
which immediately yields the lemma statement.
\end{proof}

The method presented in the proof of Lemma \ref{lem:bad-characterization} fails to work in the case where $a=b=q-1$ because the value of $P|_{W_{g H, \vec{p}}}$ becomes the sum of the coefficients of $T^{q-1}$ and $T^{2q-2}$. So, we cover this case separately.

\begin{lemma}\label{rem:bad-special-case}
Let $\mathcal H$ be a collection of cosets $g H$ of $\F_q^\times$. Then, $X^{q-1}Y^{q-1}$ is $(\mathcal H,q)-$bad.
\end{lemma}

\begin{proof}
Let $(x,y)=\vec{p} \in \F_q^2$ be any vector. We can simplify and expand equation \eqref{eq:expansion} where $a=b=q-1$, once again using the fact that $\sum_{T \in \F_q} T^{a+i} = 1$ if $a+i = q-1$ and $0$ otherwise and that $\alpha,g \neq 0$.
\[P|_{W_{g H, \vec{p}}} = \sum_{\alpha \in H} 1+(y-g\alpha x)^{q-1}\]
Then, if we pick $x = g^{-1}$ and $y = 0$, we see that $1+(y-g\alpha x)^{q-1} = 0$ if $\alpha \neq 1$ and $1$ if $\alpha = 1$. Therefore, because we are in characteristic $2$, this sum is $1$, so $X^{q-1}Y^{q-1}$ is $(\mathcal H, q)$-bad because there exists a choice of $\vec{p}$ where $P|_{W_{g H, \vec{p}}} \neq 0$.
\end{proof}

Then, because $a,b=q-1$ satisfies both conditions of Lemma \ref{lem:bad-characterization}, its conclusion holds for all choices of $a,b$.

\begin{corollary} \label{cor:bad-char-first-half}
Let $H\leq \F_q^\times$ be a subgroup and let $\mathcal H$ be the collection of cosets $g H$ of $\F_q^\times$. Then, for $a,b \leq q-1$, a monomial $X^aY^b$ is $(\mathcal H,q)$-bad if and only if both of the following conditions hold:
\begin{enumerate}
    \item $a\vee b = q-1$
    \item There exists an $i \equiv b \mod |H|$ such that $i \leq_2 a\wedge b$.
\end{enumerate}
\end{corollary}
We can immediately apply this result with a naive bound on the number of $a,b$ such that those two conditions hold. This demonstrates that wedge-lifted codes are no worse than the non-algebraic constructions of \cite{FVY15, BE16}.
\begin{corollary}
  Let $H\le \mathbb{F}_q^\times$ be a subgroup and let $\mathcal{H}$ be the collection of cosets $gH$ of $\mathbb{F}_q^\times$.
  Then, for $t=(q-1)/|H|$, the $(\mathcal{H},q)$ wedge-lifted code has the $t$-DRGP and redundancy at most $t\sqrt{N}$.
\label{cor:fvy-0}
\end{corollary}
\begin{proof}
    We bound the number of $(\mathcal H, q)$-bad monomials $X^aY^b$ by fixing $b-i \equiv 0 \mod |H|$ and bounding the number of $a,b$ that satisfy $i \leq_2 a\wedge b$ and $a\vee b = q-1$. Using \ref{cor:bad-char-first-half} gives an upper bound for the number of possible bad monomials $X^aY^b$, and even counts some monomials multiple times for different choices of $i$.
    
    Because $b \leq q-1$, there are $(q-1)/|H|$ different choices of $b-i \equiv 0 \mod |H|$. Let $q=2^\ell$ and write $b-i = \overline{x_{\ell-1}\dotsm x_0}$ and $a,b,i$ similarly, indexed by $a_j,b_j,i_j$. Then, because $i \leq_2 b$, $x_j = b_j - i_j$. There are two possibilities for each $x_j$, and we count the number of choices of $a_j,b_j$ in each case. First, if $x_j = 1$, then it must be the case that $b_j = 1, i_j = 0$, so $a_j$ can be either $0,$ or $1$. If $x_j = 0$, then it must be the case that $b_j = i_j$, and in either case $a_j = 1$ is forced because of $a\vee b = q-1$ and $i \leq_2 a$. So, in either case of $x_j$, we have two choices for $(a_j,b_j)$, so for any $b-i \equiv 0 \mod |H|$ there are $2^\ell=q$ choices of $a,b$ that satisfy $a\vee b = q-1$ and $i \leq_2 a\wedge  b$. 
    
    Therefore, because there are $(q-1)/|H|$ choices for $b-i$, there are at most $q(q-1)/|H|$ bad monomials. Observation \ref{prop:repaircount} tells us that there are $t = (q-1)/|H|$ disjoint repair groups, and because the code has length $N=q^2$, Observation \ref{lem:rate} gives us $t\sqrt{N}$ redundancy.
\end{proof}
\begin{corollary}
  \label{cor:fvy}
  For all $\alpha\in(0,1/2)$, there exists infinitely many $N=q^2$ such that there exists an $(\mathcal{H},q)$ wedge-lifted code with the $t$-DRGP for $t=N^{\alpha+o(1)}$, and redundancy at most $t\sqrt{N}$.
\end{corollary}
\begin{proof}
  We first prove when $\alpha\in(0,1/2)$ is a dyadic rational number, i.e., a rational of the form $a'/2^{b'}$ for integers $a',b'$.
  Let $a/2^b = 1-2\alpha$.
  Let $a=\overline{a_{b-1}\dots a_0}$, and $n$ be an arbitrary positive integer.
  Let $q=2^{2^bn}$ and $h=\prod_{i:a_i=1}^{} (2^{2^in}+1) = \Theta(2^{an})$.
  In this way, $h | q-1$, so there exists a subgroup $H\le \mathbb{F}_q^\times $ of size $h$.
  By Corollary~\ref{cor:fvy-0}, the $(\mathcal{H},q)$ wedge-lifted code has the $t$-DRGP with redundancy at most $t\sqrt{N}$, for $t = (q-1)/h = \Theta(2^{2^bn-an}) = \Theta(q^{2\alpha}) =  \Theta(N^{\alpha})$, as desired.
  When $\alpha$ is not a dyadic rational, take an approximation of $\alpha$ with dyadic rationals $\alpha_1,\alpha_2,\dots,$ with $\lim_{i\to\infty}\alpha_i=\alpha$, which is possible since dyadic rationals are dense in the reals, and consider the above constructions for parameter $\alpha_i$ with increasing values of $N$.
\end{proof}

\subsection{Instantiations}\label{ssec:instantiations}
In this section, we continue to assume that $\mathcal H$ is a collection of cosets $gH$, examining the special case where $\ell = \ell'\ddd$ for some natural numbers $\ddd$ and  $\ell'$, and each coset in $\mathcal H$ has order $|H| = (q - 1)/ ({q}^{1/\ddd} - 1)$. Under these conditions, we are able to give a precise description of the monomials that fail to satisfy Corollary \ref{cor:bad-char-first-half}.

We make use of two observations. The first is that finding the difference between the integers $a$ and $b$, where $a$ lies in the 2-shadow of $b$, is as simple as taking the difference between corresponding bits.
\begin{observation}
Let $a$ and $b$ be nonnegative integers with binary representations $\overline{a_{\ell-1}\dotsm a_0}$ and  $\overline{b_{\ell-1}\dotsm b_0}$ respectively. If $a \leq_2 b$, then the integer $c = b - a$ has binary representation  $\overline{c_{\ell-1}\dotsm c_0}$  where $c_i = b_i - a_i$ for every integer $i \in \{0, \dots, \ell - 1\}$.
\end{observation}

The second is that multiples of $(q - 1)/ (q^{1/\ddd} - 1)$ that are less than $q$ have repeating bit patterns.
\begin{observation}
Let $q=2^{\ell'\ddd}$ for some natural numbers $\ell'$ and $d$. Then any nonnegative integer $a < q$ with binary representation $\overline{a_{\ell'\ddd-1}\dotsm a_0}$ is a multiple of $(q - 1)/ (q^{1/\ddd} - 1)$ if and only if $\overline{a_{\ell' -1}\dotsm a_{0}} = \overline{a_{2\ell' -1}\dotsm a_{\ell'}} = \dots = \overline{a_{\ddd \ell' - 1}\dotsm a_{(\ddd -1)\ell'}}$.
\end{observation}

\begin{lemma}\label{lem:equiv-bad-char}
Let $q=2^{\ell'\ddd}$ for some natural numbers $\ell'$ and $d$. Let $H \leq \F_q^\times$ be a subgroup of order $|H| = (q - 1)/ (q^{1/d} - 1)$ and $\mathcal H$ be the collection of cosets $gH$.
Then any monomial $X^aY^b$ with $0 \leq a, b \leq q - 1$ is $(\mathcal H, q)$-bad if and only if both of the following conditions hold:
\begin{enumerate}
\item $a \vee b = q - 1$
\item For every $r,s \in \{0, 1 \dots, d - 1\}$, there does not exist $j \in \{0, 1, \dots, \ell' -1  \}$ such that $b_{r\ell' + j} = a_{s\ell' + j} = 1$ and $a_{r\ell' + j} = b_{s\ell' + j} = 0$.
\end{enumerate}
\end{lemma}

\begin{proof}
We first prove the forward direction. If $X^aY^b$ is $(\mathcal H, q)$-bad, then there exists $i \equiv b \mod |H|$ such that $i \leq_2 a \wedge b$ by Corollary \ref{cor:bad-char-first-half}. Since $i \leq_2 b$ and $b - i \equiv 0 \mod |H|$, we know that if there exist $r,s \in \{0, 1 \dots, d - 1\}$ and $j \in \{0, 1, \dots, \ell' -1 \}$ such that $b_{r\ell' + j} = 1$ and  $b_{s\ell' + j} = 0$, then $i_{r\ell' + j} = 1$ and $i_{s\ell' + j} = 0$. 
Note that $i \leq_2 a$ and $i_{r\ell' + j} = 1$ together imply that $a_{r\ell' + j} = 1$.  Note also that $a \vee b = q - 1$ and $b_{s\ell'+j} = 0$ together imply that $a_{s\ell' + j} = 1$. Hence, it is never the case that $b_{r\ell' + j} = a_{s\ell' + j} = 1$ and $a_{r\ell' + j} = b_{s\ell' + j} = 0$.

We now show that the converse also holds. Suppose that for every $r,s \in \{0, 1 \dots, d - 1\}$, there does not exist $j \in \{0, 1, \dots, \ell' -1  \}$ such that $b_{r\ell' + j} = a_{s\ell' + j} = 1$ and $a_{r\ell' + j} = b_{s\ell' + j} = 0$. Then we can construct a whole number $i \equiv b \mod |H|$ such that $i \leq_2 a \wedge b$. For each $j \in \{0, 1,\dots, \ell' - 1 \}$ we assign the values of  $i_j, i_{n + j}, \dots, i_{(d - 1)\ell' + j}$ according to the following rules:
\begin{enumerate}
\item If $b_{j} = b_{\ell' + j} = \dots =  b_{(d -1)\ell' + j}$, then let $i_{j} = i_{\ell' + j} = \dots =  i_{(d -1)\ell' + j} = 0$.
\item If there exist $r, s$ such that $b_{r\ell' + j} \neq b_{s\ell' + j}$, then let $i_{j} = b_{j}, i_{\ell' + j} = b_{\ell' + j}, \dots,  i_{(d - 1)\ell' + j} = b_{(d - 1)\ell' + j}$.
\end{enumerate}
Note that our construction ensures that $i \leq_2 b$. Note also that for every $r, s \in \{0, 1, \dots, d - 1\}$ and every $j \in \{0, 1, \dots, \ell' -1 \}$ we have $b_{r\ell' + j} - i_{r\ell' + j} = b_{s\ell' + j} - i_{s\ell' + j}$, which implies that $b - i \equiv 0  \mod |H|$. It remains to show that $i \leq_2 a$. We need not focus on the bits in $i$ that are zero, as those will satisfy the $2$-shadow condition. We are only concerned about the case that there exist $r \in \{0, 1, \dots, d- 1\}$ and $j \in \{0, 1, \dots, \ell' -1 \}$ such that $a_{r\ell' + j} = 0$ and $i_{r\ell' + j} = 1$. If such $r$ and $j$ exist, then we deduce from the rules of our construction that $b_{r\ell' + j} = 1$ and there exists $s \in \{0, 1, \dots, d - 1\}$ such that $b_{s\ell' + j} = 0$. In order to ensure that $a \vee b = q - 1$, we must have $a_{s\ell' + j} = 1$. However, this contradicts our original assumption that there does not exist $j \in \{0, 1, \dots, \ell' -1  \}$ such that $b_{r\ell' + j} = a_{s\ell' + j} = 1$ and $a_{r\ell' + j} = b_{s\ell' + j} = 0$.
\end{proof}

With the description of $(\mathcal H, q)$-bad monomials in hand, we can easily give an exact count.
\begin{cor}\label{cor:exact-count}
Let $q=2^{\ell'\ddd}$ for some natural numbers $\ell'$ and $d$. Let $H \leq \F_q^\times$ be a subgroup of order $|H| = (q - 1)/ (q^{1/d} - 1)$ and $\mathcal H$ be the collection of cosets $gH$. Then, there are $(2^{\ddd + 1} - 1)^{\ell'}$ $(\mathcal H, q)$-bad monomials.
\end{cor}
\begin{proof}
For each $j \in \{0, \dots, \ell' - 1\}$, we count the ways of assigning values to $(a_i, b_i)$ for every $i \in I = \{ i : 0 \leq i \leq \ddd \ell' -1 \ \text{and} \ i \equiv j  \mod \ell' \}$  so that the monomial $X^aY^b$ satisfies both conditions of Lemma \ref{lem:equiv-bad-char}. To satisfy the first condition, we need $(a_i, b_i) \in \{(0, 1), (1, 0), (1, 1) \}$ for each $i \in I$. To satisfy the second condition, we must not simultaneously have  $(a_i, b_i) = (0, 1)$ and $(a_{i'}, b_{i'}) = (1, 0)$ for some $i, i' \in I$. There are $2^\ddd$ ways of assigning the values of $(a_i, b_i)$ using the set $\{(0, 1), (1, 1) \}$. Likewise, there are $2^\ddd$ ways to make the assignment using the set $\{(1, 0), (1, 1) \}$. With the exception of the assignment that only uses the set $\{(1, 1)\}$, which is counted twice, we have counted every assignment once. Hence, there are $2^{\ddd+ 1} - 1$ possible assignments for each $j$. Since the assignments for each of the $\ell'$ possible values of $j$ can be done independently, we conclude that there are $(2^{\ddd + 1} - 1)^{\ell'}$ bad monomials.
\end{proof}

To conclude this section, we summarize the properties of the \codenamens{s} constructed from our special choice of $\mathcal H$.
\begin{theorem}\label{thm:q-ary-wedge}
Let $q=2^{\ell'\ddd}$ for some natural numbers $\ell'$ and $d$. Let $H \leq \F_q^\times$ be a subgroup of order $|H| = (q - 1)/ (q^{1/d} - 1)$ and $\mathcal H$ be the collection of cosets $gH$. Then, the $(\mathcal H, q)$ \codename has the following properties.
\begin{itemize}
    \item The length of the code is $q^2$.
    \item The alphabet size is $q$.
    \item The redundancy is at most $(2^{\ddd + 1} - 1)^{\ell'}$.
    \item The code has the $(2^{\ell'} - 1)$-disjoint repair group property.
\end{itemize}
\end{theorem}

\section{Trace of Wedge-Lifted Codes}
\label{sec:trace}

In this section, we take the trace of our codes in Theorem~\ref{thm:q-ary-wedge} and set parameters, in order to prove our main theorem, Theorem~\ref{thm:main}.

We begin with a lemma that says that we can take the coordinate-wise trace of a code over a field of characteristic two without hurting the redundancy or locality.

\begin{lemma}\label{lem:trace-code}
Let $\mathcal C \subseteq \F_q^{q^2}$ be a $(\mathcal H,q)$-wedge-lifted code. Then, there exists a binary code $\mathcal C' \subseteq \F_2^{q^2}$ with the same or lower redundancy and the same number of disjoint repair groups.
\end{lemma}

\begin{proof}
For any polynomial $P$ such that $\eval(P) \in \mathcal C$, $H \in \mathcal H$, and $(x,y) = \vec{p} \in \F_q^2$, we know that $P|_{W_{H,\vec{p}}} = 0$. Therefore, using the fact that $\tr_2$ commutes with addition,
\[\tr_2(0) = \tr_2\left(\sum_{\alpha \in H}\sum_{T \in \F_q} P(T,\alpha(T-x) + y)\right) = \sum_{\alpha \in H}\sum_{T \in \F_q} \tr_2(P(T,\alpha(T-x) + y))\]
Therefore, if we view the code $\tr_2(\mathcal C)$ as a code over $\F_2^{q^2}$ with coordinates indexed by $\F_q^2$, any index $\vec{p}$ of a codeword $\eval(\tr_2\circ P)$ of $\tr_2(\mathcal C)$ can be repaired by summing over the lines whose slopes are in $H$. So, $\tr_2(\mathcal C)$ and $\mathcal C$ both have $|\mathcal H|$ disjoint repair groups.

Furthermore, Theorem \ref{thm:delsarte}  states that $\dim(\tr_2(\mathcal C)) \geq \dim(\mathcal C)$, so because these codes have the same length, $\tr_2(\mathcal C)$ has the same or lower redundancy as $\mathcal C$.
\end{proof}

Finally, we are in a position to prove Theorem~\ref{thm:main}.  We restate the theorem for the reader's convenience. 

\begin{theorem*}[Theorem~\ref{thm:main}, restated]
For positive integers $\ddd$ and infinitely many $N$, for $t=N^{1/2\ddd}$, there exist binary codes of length $N$ with redundancy $t^{\log_2(2-2^{-\ddd})}\sqrt{N}$ that have the the $(t-1)$-DRGP.
\end{theorem*}

\begin{proof}[Proof of Theorem \ref{thm:main}]
 We apply Theorem~\ref{thm:q-ary-wedge}.  In particular, we choose $t = 2^{\ell'}$, and choose the $\ddd$ in the statement of Theorem~\ref{thm:main} to be the same as the $\ddd$ in the statement of Theorem~\ref{thm:q-ary-wedge}. Then $q = 2^{\ell' d} = t^d$, and the length of the resulting code is $N = q^2 = t^{2d}$.  The redundancy is 
\begin{align*}(2^{d+1}-1)^{\ell'} &= 2^{d \cdot \log_2 t} \cdot ( 2 - 2^{-d})^{\log_2 t} \\
&= \sqrt{N} \cdot (2 - 2^{-d})^{\log_2(t)} \\
&= \sqrt{N} \cdot t^{\log_2(2 - 2^{-d})},
\end{align*}
which proves the theorem.
\end{proof}

As a corollary of Lemma~\ref{lem:trace-code}, we can also make the codes of Corollary~\ref{cor:fvy}, which match \cite{FVY15} for a dense collection of $\alpha\in(0,1/2)$, into binary codes.  It would be very interesting if we could get a quantitative result like that of Theorem~\ref{thm:main} for a dense collection of $\alpha \in (0,1/2)$.
\begin{theorem}
  For all $\alpha\in(0,1/2)$, there exists infinitely many $N$ such that there exists a binary code, obtained by taking the trace code of a wedge-lifted code, with the $t$-DRGP for $t=N^{\alpha+o(1)}$, and redundancy at most $t\sqrt{N}$. 
  \label{thm:fvy}
\end{theorem}
\begin{proof}
  Apply Lemma~\ref{lem:trace-code} to Corollary~\ref{cor:fvy}.
\end{proof}

\section{Conclusion}
\label{sec:conclusion}
In this paper, we introduced \emph{wedge-lifted codes}, which give an improved construction of binary codes with the $t$-DRGP for several $t \leq \sqrt{N}$.  
We conclude with some open questions.

\begin{enumerate}
\item For all $\alpha\in(0,1/2)$ with $t=N^{\alpha}$, are there binary codes with the $t$-DRGP and redundancy $O(t^{1-\varepsilon}\sqrt{N})$ for some $\varepsilon>0$, possibly (but ideally not) depending on $\alpha$?
Our work shows this is true for $\alpha=1/2\ddd$ when $\ddd$ is any positive integer.
The work of \cite{LW19} showed this is true (with an absolute $\varepsilon=0.425$) for nonbinary codes.

\item Can one improve \cite{RV16, Woo16} to prove better lower bounds on the redudancy of codes with the $t$-DRGP?
\end{enumerate}

\bibliographystyle{alpha}
\bibliography{bib}

\end{document}